\RequirePackage{fix-cm}
\documentclass[smallextended]{svjour3}       
\smartqed  
\usepackage[english]{babel}
\usepackage[utf8]{inputenc}
\usepackage{ulem}
\usepackage[round]{natbib} 
\usepackage{url}
\usepackage{appendix}
\usepackage{subcaption}
\usepackage{graphicx}
\usepackage{multirow} 
\usepackage{float} 
\usepackage{rotating} 
\usepackage{algorithm,algorithmic}
\usepackage{mathtools}
\usepackage{dsfont} 
\usepackage{amsmath}

\usepackage{amsthm} 
\usepackage{amssymb} 

\usepackage{arydshln} 
\usepackage{cancel} 
\newcommand{\Cov}{{\mathrm{Cov}}}

\renewcommand{\P}{{\mathcal{P}}}
\newcommand{\T}{{\mathcal{T}}}
\newcommand{\Q}{{\mathcal{Q}}}

\newcommand{\HypNull}{{\mathcal{H}_0} }
\newcommand{\HypAlt}{{\mathcal{H}_1} }

\newcommand{\clOne}{C^1 }
\newcommand{\clTwo}{C^2 }

\newcommand{\CvarONE}{C^1_{ij} }

\newcommand{\cVarTWO}{C^2_{ij} }
\newcommand{\VarTriplets}{C^1_{ij}C^2_{ij'} }

\newcommand{\cvarREALONE}{c^1_{ij} }
\newcommand{\cVarREALTWOPrime}{c^2_{i'j'} }
\newcommand{\cVarREALTWO}{c^2_{ij} }

\newcommand{\VarREALPairs}{c^1_{ij}c^2_{ij} }
\newcommand{\VarREALTriplets}{c^1_{ij}c^2_{ij'} }
\newcommand{\VarREALQuadruplets}{c^1_{ij}c^2_{i'j'} }

\newcommand{\RIOld}{RI}
\newcommand{\ARIOld}{ARI}

\newcommand{\pikl}{\pi_{k\ell}}
\newcommand{\pik}{\pi_{k.}}
\newcommand{\pil}{\pi_{.\ell}}
\newcommand{\nkl}{n_{k\ell}}
\newcommand{\nl}{n_{.\ell}}
\newcommand{\nk}{n_{k.}}

\newtheorem{prop}{Proposition}[section] 

\begin{document}

\title{Adjusting the adjusted Rand Index - A multinomial story
}


\author{Martina Sundqvist         \and
        Julien Chiquet \and
        Guillem Rigaill 
}


\institute{M. Sundqvist
\at MIA-Paris, UMR 518 AgroParisTech, INRAE, Université Paris-Saclay, 75005, Paris, France
\at  Institut Curie - PSL Research University, Translational Research Department, Breast Cancer Biology Group, 26 rue d’Ulm, 75005 Paris, France
           \email{ma.sundqvis@gmail.com}          
\and
J. Chiquet
\at  MIA-Paris, UMR 518 AgroParisTech, INRAE, Université Paris-Saclay, 75005, Paris, France
   \email{julien.chiquet@inrae.fr}
\and
G. Rigaill 
\at Université Paris-Saclay, CNRS, INRAE, Univ Evry,  Institute of Plant Sciences Paris-Saclay (IPS2), 91405, Orsay, France
\at  Université de Paris, CNRS, INRAE, Institute of Plant Sciences Paris-Saclay (IPS2), 91405, Orsay, France
\at Laboratoire de Mathématiques et Modélisation d'Evry (LaMME), Université d'Evry Val d’Essonne, UMR CNRS 8071, ENSIIE, USC INRA, guillem.rigaill@inra.fr
}


\date{Received: date / Accepted: date}

\maketitle

\begin{abstract}
The Adjusted Rand Index ($ARI$) is arguably one of the most popular measures for cluster comparison. The adjustment of the $ARI$ is based on a hypergeometric distribution assumption which is unsatisfying from a modeling perspective as (i) it is not appropriate when the two clusterings are dependent, (ii) it forces the size of the clusters, and (iii) it ignores randomness of the sampling. In this work, we present a new "modified" version of the Rand Index. First, we redefine the $MRI$ by only counting the pairs consistent by similarity and ignoring the pairs consistent by difference, increasing the interpretability of the score. Second, we base the adjusted version, $MARI$, on a multinomial distribution instead of a hypergeometric distribution. The multinomial model is advantageous as it does not force the size of the clusters, properly models randomness, and is easily extended to the dependant case. We show that the $ARI$ is biased under the multinomial model and that the difference between the $ARI$ and $MARI$ can be large for small $n$ but essentially vanish for large $n$, where $n$ is the number of individuals. Finally, we provide an efficient algorithm to compute all these quantities ($(A)RI$ and $M(A)RI$) by relying on a sparse representation of the contingency table in our \texttt{aricode} package. The space and time complexity is linear in the number of samples and importantly does not depend on the number of clusters as we do not explicitly compute the contingency table.

\keywords{ Clustering \and Rand Index \and Multinomial distribution \and Statistical Inference}
\end{abstract}
\section{Introduction}\label{sec:Intro}

With the increasing amount of data available, development of clustering methods have become crucial in unsupervised learning to explore and find patterns in data sets. 
Despite the  wealth of theoretical research on this subject, in practice selecting and validating a clustering is difficult. To answer these questions, one often resorts to a measure of clustering comparison: when the data is labeled, the quality of the clustering is evaluated by measuring the overlap with the original labeling; in the absence of labels, the reliability of the clustering can be assessed by evaluating its stability \citep[see, e.g.][]{von2010clustering}. This can be done by comparing several clusterings obtained by perturbing the initial data set (i.e. with resampling), or by running different clustering methods on the same data set. The idea of clustering stability is dug deeper in cluster ensembles \citep{strehl2002cluster} and its variants,  which involve measures of clustering comparison in the construction of the  clustering itself.

Among the many measures proposed for pairwise clustering comparisons \citep[see][for an overview]{vinh2010information} one of the most popular is the Rand index ($RI$) \citep{rand1971} and its adjusted variant \citep{hubert1985comparing, morey1984measurement}.  The $RI$ is designed to estimate the probability of having a coherent pair, which is a pair for which its two observations are either in the same group in the two compared clusterings or in different groups. It is computed from the contingency table of the two classifications.  However, the $RI$ depends on the number of groups \citep{morey1984measurement} and  is therefore difficult to interpret. To overcome this issue, the Adjusted Rand Index (in short $ARI$) is obtained by subtracting to the $RI$ an estimator of its expected value obtained under the assumption of two independent clusterings. 

To obtain such an estimator, a population distribution has to be assumed upon the two compared clusterings, or more specifically upon the marginals of the contingency table of the two clusterings. Considering either the clusters sizes fixed or not, the two natural hypotheses that arise are either the hypergeometric distribution or the multinomial distribution. In the literature, there is discordance as to which of these hypotheses to use.

The $RI$ and $ARI$ as defined by \cite{brennan1974measuring} and then adapted by \cite{hubert1985comparing} are based on the hypergeometric distribution hypothesis. In fact, considering fixed cluster sizes makes calculations easier and the expected value of the $RI$ deterministic. However, this is a strong assumption that is violated in all cluster studies since no clustering algorithm fixes cluster sizes \citep[see][for a detailed discussion]{wagner2007comparing}. Moreover, from a modeling perspective, it implicitly ignores any randomness of the sampling procedure and considers that the set of individuals that we observed is fixed. Hence under this model the $(A)RI$ are post-hoc quantities for which no inference to a parental population can be done, which limits the interpretation exclusively to the observed data points. Assuming the marginal to be fixed certainly simplifies the calculations under the hypothesis of independence between clustering. However, modeling dependency between clusterings under this assumption is not straightforward and rather unnatural compared to the multinomial model. Yet one certainly hopes to compare clusterings that are alike or dependant.

In comparison, the multinomial model does not assume the size of the clusters to be fixed, by considering a sample observed from an infinite population. Modeling dependent clusterings and adjusting accordingly is then greatly simplified. For all these reasons we argue that the multinomial model is more natural from a statistical perspective. Note that \cite{morey1984measurement} already studied this model to propose an adjusted version of the $RI$. Nonetheless, as pointed out in \cite{hubert1985comparing, steinley2004properties,steinley2018note}, \citeauthor{morey1984measurement} made an error in their calculation of the expected value of the $RI$, assuming that the expected value of a squared variable is the square of the expected value, which is wrong in general. We are convinced that this error is the reason for the problem described in \cite{steinley2018note}, advocating unfairly for the hypergeometric version of the $(A)RI$.

\begin{center}
\S
\end{center}

In this work, we essentially make a rigorous statistical analysis of the $RI$ under the hypothesis of a multinomial distribution. In details, our contributions are the following:
\begin{enumerate}
    \item Define new versions of the $RI$ and the $ARI$, denoted by $MRI$ and $MARI$ (for "modified" $(A)RI$), only counting consistent pairs by similarity. Indeed, we show that counting consistent pairs by dissimilarity is unnecessary and blurs the interpretation. In terms of our newly defined $MARI$, considering those pairs would simply result in a multiplication by 2. 
    \item Finalise the work of \cite{morey1984measurement} and derive an unbiased estimator of the expected value of the $MRI$ under a multinomial distribution valid for data under  $\HypAlt$ (dependent clusterings) and $\HypNull$ (independent clusterings).
    \item Provide an efficient algorithm to compute all these quantities ($(A)RI$ and $M(A)RI$) by relying on a sparse representation of the contingency table. The complexity is in $\mathcal{O}(n)$ time and space where $n$ is the number of individuals. This is better than the usual  $\mathcal{O}(n + KL)$ complexity, where  $K$ and $L$ are the sizes of the two clusterings one which to compare, typically obtained when using the non-sparse contingency table. Our code is available in versions $\geq 1.0.0$ of the \texttt{R} package \texttt{aricode} \citep{aricode}.
    \item Investigate the difference with the hypergeometric \citeauthor{hubert1985comparing}'s $ARI$ and show that it is biased under the multinomial distribution, even if the difference between the two estimators remains small. This is in contradiction with the results of \cite{steinley2018note} that used the faulty $ARI$ of \cite{morey1984measurement}.
\end{enumerate}

\section{Statistical Model}\label{sec:modelStat}

\subsection{A new Rand Index - counting only pairs consistent by similarity}\label{sec:newdefRI}

The Rand Index ($RI$) proposed by \cite{rand1971} counts all the consistent pairs in two given classifications. In details, let us consider two classifications $\clOne$ and $\clTwo$ in respectively $K$ and $L$ classes of the same $n$ individuals. The labels of individual $i$ are given by 
$c^1_{i} \in [1, \ldots ,K]$ and $c^2_{i} \in  [1, \ldots, L]$.  The consistent pairs are all pairs where observations $i$ and $j$ are in the same group  (consistent by similarity), or in different groups (consistent by difference) in $\clOne$ and $\clTwo$. 

We introduce the two quantities $c^1_{ij}$ and $c^2_{ij}$ indicating whether $i$ and $j$ are in the same group for respectively classification $\clOne$ and $\clTwo$ : 
\begin{equation*}
c^1_{ij} = \left\{
    \begin{array}{ll}
        1 & \text{if }  c^1_{i} = c^1_{j} = k, \\
        0 & \text{otherwise},
    \end{array}
\right. \quad \text{and}
\quad c^2_{ij} = \left\{
    \begin{array}{ll}
        1 & \text{if }  c^2_{i} = c^2_{j} = \ell, \\
        0 & \text{otherwise}.
    \end{array}
\right.    
\end{equation*}
Note that $c^1_{ij}$ and $c^2_{ij}$ are the realisations of Bernoulli random variables denoted by $\CvarONE$ and $\cVarTWO$ that will prove useful later in our statistical analysis, while studying the $RI$ and other similar quantities as random variables.

Using these two quantities we see that a pair is consistent by similarity if $c^1_{ij} c^2_{ij} = 1$ and consistent by difference if $(1-c^1_{ij})(1-c^2_{ij})= 1$. 
Now considering all pairs, we get the following formula for the $RI$ as defined by \citeauthor{rand1971}:
\begin{equation}
    \begin{split}
        \RIOld(\clOne,\clTwo) & = \frac{1}{{n\choose{2}}} \sum_{i<j} c^1_{ij}c^2_{ij} + \sum_{i<j}(1- c^1_{ij})(1- c^2_{ij}) \\
        & = 1 + \frac{1}{{n\choose{2}}}\bigg[2\sum_{i<j} c^1_{ij}c^2_{ij} - \sum_{i<j} c^1_{ij} - \sum_{i<j} c^2_{ij} \bigg].
    \end{split} \label{eq:redef_justif}
\end{equation}

In Equation \eqref{eq:redef_justif}, we remark that only the product $\sum c^1_{ij}c^2_{ij}$ depends on the joint distribution of $\clOne$ and $\clTwo$: all other terms, coming exclusively from coherent pairs by difference, depend on the marginal distributions of $\clOne$ and $\clTwo$. These terms will thus be cancelled out in any adjusted version of the $\RIOld$, correcting for what would happen if $\clOne$ and $\clTwo$ were drawn independently. 
Hence, we argue that considering the consistent pairs by difference unnecessarily complicates the reasoning and the probabilistic analysis of the $RI$. For simplicity we thus redefine the index and refer to it as the $MRI$ (for "modified" $RI$):

 \begin{equation}\label{eq:RI}
    \begin{split}
        MRI(\clOne,\clTwo) & = \frac{1}{{n\choose{2}}} \sum_{i<j} c^1_{ij}c^2_{ij}.
    \end{split}
\end{equation}

\begin{remark}
For the derivation of the expected value of $MRI$, $RI$ and their adjusted version $MARI$ and $ARI$, using the definition involving $c^1_{ij}$ and $c^2_{ij}$ (or more exactly $C^1_{ij}$ and $C^2_{ij}$ in a probabilistic perspective) considerably simplify the calculations compared to their classical combinatorial formulations. These combinatorial formulations are recalled in the next section as they are classically used to compute
the $RI$ and its variants.
\end{remark}

 \subsection{Computing the Rand Index from the $\nkl$ contingency table }\label{sec:RICompute}

The information from two observed classifications is usually summarized in a contingency table like Table \ref{tab:n_kl Contigency}, representing the number of observations $\nkl$ in group $k$ in $\clOne$ and in group $\ell$ in $\clTwo$.  

\begin{table}[!htbp]
\caption{Contingency Table between clusterings $C^1$ and $C^2$;  each entry $\nkl$ corresponds to the number of observations in group $k$ in $C^1$ and group $\ell$ in $C^2$.}
\begin{center}
\begin{tabular}{c|ccccc|c}
    ${ \atop \clOne}\!\diagdown\!^{\clTwo}$ & $c^2_{1}$& $\cdots$ &  $c^2_{\ell}$ & $\cdots$ & $c^2_{L}$ & Sums \\ \hline
   $c^1_{1}$ & $n_{11}$  & $\cdots$ & $n_{1\ell}$ & $\cdots$ & $n_{1L}$ & $n_{1.}$ \\
   \vdots  & \vdots &  $\ddots$ & \vdots & $\ddots$ & \vdots & \vdots \\ 
    $c^1_{k}$ & $n_{k1}$ &  $\cdots$ &$n_{k\ell}$ &  $\cdots$ & $n_{2L}$ & $n_{k.}$ \\
\vdots & \vdots &   $\ddots$ & \vdots & $\ddots$ & \vdots & \vdots \\ 
   $c^1_{K}$ & $n_{K1}$ & $\cdots$ & $n_{K\ell}$ & $\cdots$ & $n_{KL}$ & $n_{K.}$ \\\hline
   Sums & $n_{.1}$ &  $\cdots$ & $n_{.\ell}$ & $\cdots$ & $n_{.L}$ & $\sum_{k\ell} \nkl=n$\\
  \end{tabular}
 \end{center}
\label{tab:n_kl Contigency}
\end{table}

Using basics combinatorics we get the following relations between $n_{k\ell}, n_{k.}, n_{.\ell}$ and $c_{ij}^1,c_{ij}^2$:
\begin{equation}\label{eq:cVarCounts}
    \sum_{i < j} c_{ij}^1= \sum_k {n_{k.} \choose 2}, \quad \sum_{i < j} c_{ij}^2 = \sum_\ell {\nl \choose 2} \text{  and  } \sum_{i < j} c_{ij}^1 c_{ij}^2 = {\nkl \choose 2}. 
\end{equation}
Expressions \eqref{eq:redef_justif} and \eqref{eq:RI} of $RI$ and $MRI$ turn to 
 \begin{eqnarray}
 \label{eq:formRIM}
    MRI(\clOne,\clTwo) & = & \frac{1}{{n\choose{2}}} \sum_{k,\ell} {\nkl\choose{2}} = \frac{1}{2{n\choose{2}}}\sum_{k,\ell} (n_{k\ell}^2 - n) \\[1.5ex] 
    \RIOld(\clOne,\clTwo) & = & 1 + \frac{2}{{n\choose{2}}} \sum_{k, l} {\nkl\choose{2}} - \frac{1}{{n\choose{2}}}\bigg[\sum_{k}{n_{k.}\choose{2}} + \sum_{l}{n_{.l}\choose{2}} \bigg]. \label{eq:formRIOld} 
\end{eqnarray}

Using these formula, one can see that the minimum of the $MRI$ is obtained when all $n_{k\ell}$ are equal, which has a simple and straightforward interpretation (as two perfectly independent and balanced clusterings). On the other hand the minimum of the $\RIOld$ is obtained for an extremely unbalanced table, \textit{i.e.} when one of the two clustering consists of a single cluster and the other only of clusters containing single points. This makes the interpretation of the $RI$ rather difficult (i.e. the lowest value is not obtained for two perfectly independent and balanced clusterings) and give more credibility to the definition of $MRI$ that does not consider consistent pairs by difference.

\subsection{Probabilistic model and properties of the Rand Index}\label{sec:probabilisticRI}

So far, the $(M)RI$ have been computed from the \textit{observed quantities} $c^1_{ij}, c^2_{ij}$, or equivalently from the observed contingency table $n_{k\ell}$. From now, we aim to study the statistical properties of the $MRI$ and consider its status of random variable\footnote{By a slight abuse of notation, we use $MRI$ for both its observed value and its definition as a random variable. We think that the context suffices for the reader to remove any ambiguity.}:
\begin{equation}
\label{eq:MRI_prob}
    MRI(\clOne,\clTwo) = \frac{1}{{n\choose{2}}} \sum_{i<j} C^1_{ij} \, C^2_{ij},
\end{equation}
where we recall that $C^1_{ij}$ and $C^2_{ij}$ are Bernoulli random variables indicating whether individual $i$ and $j$ are in the same groups in classification $C^1$ respectively $C^2$.

To derive the probability of success associated to  $C^1_{ij}$ and $C^1_{ij}$, we need a probabilistic model for the classification of a given individual in $C^1$ and $C^2$, that is, a counterpart for generating the two observed  clusterings $c^1_i$ and $c^2_i$ for the $n$ data points. We denote by $C^1_i$ and $C^2_i$ the corresponding random variables. A natural model is the multinomial model, which give the joint distribution of $(C^1_i, C^2_i)$ as follows: for all $(k,\ell) \in \{1,\dots,K\} \times  \{1,\dots,L\}$,
\begin{equation*}
    \mathbb{P}(C^1_i = k,  C^2_i = \ell) = \pikl, \quad \text{s.c. } \sum_{k,\ell}^{K,L}\pikl = 1.
\end{equation*}
The marginal probabilities of a given group is defined for $k$ in $\clOne$ by $ \sum_{\ell}^{L}\pikl = \pi_{k.}$ and for $\ell$ in $\clTwo$ by $\sum_{k}^{K}\pikl = \pil$. See Table \ref{tab:Pi Contigency} for a global picture. Compared to the hypergeometric model, the multinomial model easily deals with dependent classifications and does not force the size of the clusters.

\begin{table}[!htbp]
\caption{Multinomial model: probabilistic distributions $\pikl = \mathbb{P}(C^1_i = k,  C^2_i = \ell)$ and marginal distributions $\pi_{k.} = \mathbb{P}(C^1_i = k)$ and $\pi_{.\ell} = \mathbb{P}(C^2_i = \ell)$}
\begin{center}
\begin{tabular}{c|ccccc|c}
    ${ \atop \clOne}\!\diagdown\!^{\clTwo}$ & $c^2_{i}=1$& $\cdots$ &  $c^2_i=\ell$ & $\cdots$ & $c^2_i = L$ & Sums \\ \hline
   $c^1_i = 1$ & $\pi_{11}$  & $\cdots$ & $\pi_{1\ell}$ & $\cdots$ & $\pi_{1L}$ & $\pi_{1.}$ \\
   \vdots  & \vdots &  $\ddots$ & \vdots & $\ddots$ & \vdots & \vdots \\ 
    $c^1_i = k$ & $\pi_{k1}$ &  $\cdots$ &$\pi_{k\ell}$ &  $\cdots$ & $\pi_{2L}$ & $\pi_{k.}$ \\
\vdots & \vdots &   $\ddots$ & \vdots & $\ddots$ & \vdots & \vdots \\ $c^1_i = K$ & $\pi_{K1}$ & $\cdots$ & $\pi_{K\ell}$ & $\cdots$ & $\pi_{KL}$ & $\pi_{K.}$ \\\hline
   Sums & $\pi_{.1}$ &  $\cdots$ & $\pi_{.\ell}$ & $\cdots$ & $\pi_{.L}$ & $\sum_{k\ell} \pikl=1$\\
  \end{tabular}
 \end{center}
\label{tab:Pi Contigency}
\end{table}

Based on this multinomial model for $C^1_i$ and $C^2_i$, it is then relatively straightforward to derive the joint distribution and marginals of $\CvarONE$ and $\cVarTWO$. In particular we have:
\begin{equation}\label{eq:basicProbsOnC}
   \mathbb{P}( \CvarONE = 1) = \sum_k \pik^2, \quad  \mathbb{P}( \cVarTWO = 1) = \sum_\ell \pil^2 \ \text{and}   \quad  \mathbb{P}( \CvarONE\cVarTWO 
   = 1) = \sum_{k,\ell} \pikl^2.
\end{equation}
However, in order to derive the expectation, variance and unbiased adjustment of the $MRI$ under the multinomial model, one not only needs to characterize events on the classification $C^1$  and $C^2$ on (unordered) pairs of individual  $\{i,j\}$, but also on \textit{pairs of pairs} of individual $\{i,j\}$ and $\{i',j'\}$, with terms like the expectation of $C^1_{ij} \times C^2_{i'j'}$. The following section derives a couple of technical -- yet simple -- lemmas, on events implying such random variables so that the final calculation of the moments of $MRI$ under the multinomial model are straightforward.

\begin{remark}
To our knowledge most derivations of the expectation and variance of the $\RIOld$ found in the literature are based on the combinatorial formulation given in Equation \eqref{eq:formRIOld}: these derivations rely on general results on the moments of either the multinomial or the generalized hypergeometric distribution and involve tedious calculations. In contrast, our  proofs, found in the next sections, are short, self-contained and easily accessible to any reader with some basic knowledge in probability and statistics. For this reason we argue that our proofs are interesting in their own rights.
\end{remark}

\subsubsection{Subsets of Pairs of Pairs - preparing the derivations of the moments of the $MRI$}\label{sec:quadruplets}

Consider $\{i, j\}$ and $\{i', j'\}$ the $\P \times \P$ set of unordered pairs of $\{1,\dots,n\}^2$ such that $i<j$ and $i'<j'$. This set is composed by pairs of pairs, and can equivalently be seen as the set of all quadruplets of  $\{1,\dots,n\}^4$ such that  $i<j$ and $i'<j'$. We partition this set into the three following subsets:
\begin{enumerate}
    \item the unordered pairs $\P$,
    \item the ordered-triplets $\T,$
    \item  the ordered quadruplets $\Q$. 
\end{enumerate}
These three subsets $\P, \T$ and $\Q$ makes a partition of $\P \times \P$ and in particular, \[
    |\P|^2 = |\P| + |\T| + |\Q|.
\]

We now study respectively  $\P, \T$ and $\Q$ in the three following lemmas: we derive their cardinality and compute some expectations involving these subsets and the $C^1_{ij}$, $C^2_{ij}$ variables under the multinomial model. These three lemmas will be the building blocks for the characterization of the $MRI$.

\begin{lemma}[Subset of unordered pairs $\P$] \label{lemma:pairs}
With a slight abuse of notation, we consider $\P$ as a subset of $\P \times \P$:
  \begin{equation*}
    \P = \{\{i, j, i', j'\}: |\{i, j\} \cup \{i', j'\}| = 2.\}
 \end{equation*}
The cardinality of $\P$ is $|\P| = {n \choose 2}$ and 
\begin{equation}\label{eq:Pairs_expected}
\mathbb{E}\left(\sum_{i, j \in \P} \CvarONE\cVarTWO\right) = {n \choose 2} \sum_{k\ell} \pikl^2.
\end{equation}
\end{lemma}

\begin{proof}
For any $i, j  \in \P$, we have from \eqref{eq:basicProbsOnC} that $\mathbb{E}( \CvarONE\cVarTWO) = \sum_{k\ell} \pikl^2$. We just need to sum over all possible pairs to get the desired result.
\end{proof}

\begin{lemma}[Subset of ordered triplets $\T$]
\label{lemma:triplet}
Consider the subset $\T$ of $\P\times\P$
 \begin{equation*}
     \T = \{\{i, j, i', j'\}: |\{i, j\} \cup \{i', j'\}| = 3\}.
 \end{equation*}
The cardinality of $\T$ is  $|\T| = n(n-1)(n-2)$ and 

\begin{equation}
    \label{eq:Triplets_expected}
    \mathbb{E} \left( \sum_{\T} \VarTriplets\right) = n(n-1)(n-2) \sum_{k\ell} \pikl\pik\pil.
\end{equation}

\begin{equation}
    \label{eq:Triplets_expected2}
    \mathbb{E}\left( \sum_{\T} C^1_{ij} C^2_{ij} C^1_{ij'}C^2_{ij'}\right) = n(n-1)(n-2) \sum_{k\ell} \pikl^3.
\end{equation}
\end{lemma}

\begin{proof} For the cardinality of $\T$, one can map to the set of arrangements of $\{1, \ldots, n\}^3.$ 

For \eqref{eq:Triplets_expected}, remark that $C^1_{ij} C^2_{ij'}$ is a Bernoulli variable equal to 1 only when $i$ and $j$ are in the same cluster $k$ in $C^1$ and $i$ and $j'$ are in the same cluster $\ell$ in $C^2$ . Hence, $j$ can be in any cluster $\ell'$ in $\clTwo$ and $j'$ can be in any cluster $k'$ in $\clOne$. From here one easily get its expectation, 
\begin{equation*}
       \mathbb{E}(\VarTriplets) = \sum_{k\ell k'\ell'}\pikl\pi_{k'\ell}\pi_{k\ell'} = \sum_{k\ell k'\ell'}\pi_{k\ell}\sum_{k'}\pi_{k'\ell}\sum_{\ell'}\pi_{k\ell'} = \sum_{k\ell}\pikl\pil\pik
\end{equation*}
and we get the desired result by summing over all triplets.

For \eqref{eq:Triplets_expected2}, remark that $C^1_{ij} C^2_{ij} C^1_{ij'}C^2_{ij'}$ is a Bernoulli variable equal to 1 if and only if $i, j$ and $j'$ are in the same clusters for both classifications. Summing over all $\T$ we get \eqref{eq:Triplets_expected2}.
\end{proof}

\begin{lemma}[Subset of ordered quadruplets $\Q$]\label{lemma:quad}
Consider the following subset $\Q$ of $\P\times\P$:
\begin{equation*}
    \Q = \{\{i, j, i', j'\}: |\{i, j\} \cup \{i', j'\}| \}= 4\}.
\end{equation*}
The cardinality is  $|\Q| = 6 {n \choose 4}$ and
\begin{equation}\label{eq:Quadruplets_expected}
    \mathbb{E} \left( \sum_{\Q} C^1_{ij}C^2_{i',j'} \right) = 6 {n \choose 4} \sum_{k,\ell} \pik^2\pil^2.
\end{equation}
\end{lemma}
 
\begin{proof}
There are ${n \choose 4}$ ways to pick 4 distinct elements of $\{1, ..., n\}^4$. We can then arrange those in ${4 \choose 2}$ to get an element of $\Q$. Hence, all together there are $6 {n \choose 4}$ quadruplets. 
We get $\mathbb{E}( \VarTriplets)$ using the fact that $i, j, i', j'$ are all different and that their classes  are drawn independently. We then sum over $\Q$.
\end{proof}

\subsubsection{Expectation and Variance of the Rand Index}\label{sec:RIpropChar}

With Lemmas~\ref{lemma:pairs}, ~\ref{lemma:triplet} and \ref{lemma:quad}, we are now equipped to easily derive the moments of the $MRI$. We use $\mathbb{E}$ for stating the expectation understood under the multinomial model in general. With the additional assumption of independence between the classification, what we refer to as the \textit{null hypothesis}, we use $\mathbb{E}_{\mathcal{H}_0}$. This terms is classically used for adjusting the Rand index.

\begin{prop}[Expectations of the $MRI$] 
Let $\theta$ denote the expectation of the $MRI$ and $\theta_0$ the expectation under $\mathcal{H}_0$. Then, 
$$ \theta = \mathbb{E}(MRI) = \sum_{k\ell} \pikl^2  , \qquad \qquad \theta_0 = \mathbb{E}_\HypNull(MRI) = \sum_{k\ell} \pik^2 \pil^2 $$
\end{prop} 

\begin{proof} By Definition \ref{eq:MRI_prob} and Lemma \ref{lemma:pairs} we obtain $\theta$. For $\theta_0$, it suffices to replace $\pi_{kl}$ by $\pi_{k.}\pi_{.l}$ in the previous formula.
\end{proof}

Similarly, we derive the expectation of the "usual" $\RIOld$. 
\begin{prop}\label{prop:ExOldRI}
Let $\theta^{RI}$ denotes the expectation of the $RI$ and $\theta_0^{RI}$ the expectation under $\mathcal{H}_0$. Then,
\begin{eqnarray*}
\theta^{RI} & = & \mathbb{E}(\RIOld) = 1 + 2\sum_{k\ell} \pikl^2 - \sum_{k} \pik^2 - \sum_{\ell}\pil^2  \\
\theta_0^{RI} & = & \mathbb{E}_\HypNull(\RIOld)= 1 + 2\sum_{k\ell} \pik^2 \pil^2 - \sum_{k} \pik^2 - \sum_{\ell}\pil^2 
\end{eqnarray*}
\end{prop} 

\begin{proof} Compared to the $MRI$, the only additional terms are $1 + \sum_{i,j} C^1_{ij} + \sum_{i,j} C^2_{ij}$. Using \eqref{eq:basicProbsOnC} and summing over all pairs $\P$ we get the desired results.
\end{proof}

We now continue with the variance of the $MRI$.

\begin{prop} Let $\sigma^2= \mathbb{V}(MRI)$ be the variance of the $MRI$. Then, 
\begin{equation*}
\sigma^2 = 
\frac{1}{{{n}\choose{2}}} \ \left(\sum_{k,\ell} \pikl^2 - \left[\sum_{k,\ell} \pikl^2\right]^2 \right) + \frac{n(n-1)(n-2)}{{{n}\choose{2}}^2} \left( \sum_{k,\ell} \pikl^3 - \left[\sum_{k,\ell} \pikl^2\right]^2 \right)          
\end{equation*}
\end{prop}

\begin{proof}
To obtain the variance of the $MRI$, first rewrite the variance in terms of covariance:
\begin{eqnarray*}
    \sigma^2 & =  & \frac{1}{{{n}\choose{2}}^2} \ \mathbb{V} \Big( \sum_{\P\times\P} C^1_{ij} C^2_{ij}\Big) \\
    & = & \frac{1}{{{n}\choose{2}}^2} \ \Cov \Big( \sum_{\P\times\P} C^1_{ij} C^2_{ij}, \sum_{\P\times\P} C^1_{ij} C^2_{ij} \Big) \\
           & = & \frac{1}{{{n}\choose{2}}^2} \ \sum_{\P\times\P} \Cov \Big( C^1_{ij} C^2_{ij}, C^1_{i'j'} C^2_{i'j'} \Big) 
\end{eqnarray*}

We then split this final sum using our partition of $\P \times \P$. Also noticing that for all $\{i,j\}, \{i', j'\} \in \Q$ we have $\Cov(C^1_{ij} C^2_{ij}, C^1_{i'j'} C^2_{i'j'}) = 0,$ we get, 

\begin{eqnarray*}
\sigma^2 & = & \frac{1}{{{n}\choose{2}}^2} \left[ \sum_{\P} \Cov\Big(C^1_{ij} C^2_{ij}, C^1_{ij} C^2_{ij}\Big) + \sum_{\T} \ \Cov\Big(C^1_{ij} C^2_{ij}, C^1_{ij'} C^2_{ij'}\Big)\right] \\
      & = & \frac{1}{{{n}\choose{2}}} \ \mathbb{V}\Big(C^1_{ij} C^2_{i,j}\Big) + \frac{n(n-1)(n-2)}{{{n}\choose{2}}^2} \ \Cov \Big(C^1_{ij} C^2_{ij}, C^1_{ij'} C^2_{ij'} \Big) \\
      & = & \frac{1}{{{n}\choose{2}}} \ \left(\sum_{k,\ell} \pikl^2 - \left[\sum_{k,\ell} \pikl^2\right]^2 \right)+ \frac{n(n-1)(n-2)}{{{n}\choose{2}}^2} \left( \sum_{k,\ell} \pikl^3 - \left[\sum_{k,\ell} \pikl^2\right]^2 \right)
\end{eqnarray*}

We get the second line by enumerating the elements of $\P$ and $\T$. We get the third line using the definition of the covariance (for any two variable $X$ and $Y$: $\Cov(X, Y) = \mathbb{E}(XY) - \mathbb{E}(X)\mathbb{E}(Y)$) and Lemmas \ref{lemma:pairs} and \ref{lemma:triplet}.

\end{proof}

\begin{remark}
Importantly, for a fixed $\pikl$, $\sigma^2$ goes towards $0$ when $n$ grows to infinity: the larger $n$, the better the estimation of $\theta$.
\end{remark}

\subsubsection{The Rand Index depends on the number of groups}\label{sec:RIdepOnNbGrs}

In the multinomial model with uniform clusters (equal cluster size),  \cite{morey1984measurement} showed that $\theta^{RI}_0$ depends on the number of groups in $\clOne$ and $\clTwo$. This is also true for $MRI$ and easier to prove since it does not include the marginal terms of coherence by difference. We also prove the following lemma showing that if one splits a cluster of $\clOne$ or $\clTwo$ into two, the $MRI$ always decreases. Note that this latter lemma does not assume independence between classifications.

\begin{lemma} 
Consider two classifications $\clOne$ and $\clTwo$ in $K+1$ respectively $L$ clusters. Let $\clOne{'}$ be the classification obtained by fusing two clusters of $\clOne$. Then,
\[MRI(\clOne, \clTwo) \leq MRI(\clOne{'}, \clTwo).\] Also, for any distribution on $\clOne$ and $\clTwo$ we have 
\[\theta(\clOne, \clTwo) \leq \theta(\clOne{'}, \clTwo))\]
\end{lemma}

\begin{proof}
Assuming without loss of generality that clusters $1$ and $2$ were merged, we get
\begin{eqnarray*}
MRI(\clOne, \clTwo) - MRI(\clOne{'}, \clTwo) & = & \frac{1}{2{n \choose 2}} \Big(\sum_{\ell} n_{1\ell}^2 + n_{2\ell}^2 - (n_{1\ell} + n_{2\ell})^2\Big) \\
& = & -\frac{1}{{n \choose 2}} \sum_{\ell} n_{1\ell}n_{2\ell} \leq 0.
\end{eqnarray*}
Since the expectation is linear, we can consider any particular model on $\clOne$ and $\clTwo$ to get the final result.
\end{proof}

\subsection{The Adjusted version of the Rand Index }\label{sec:ARI}

Since the $(M)RI$ depends on the number of groups, it needs to be adjusted for chance. A way to do so, is to subtract its expectation under the null hypothesis $\HypNull$ \citep[as motivated in][]{brennan1974measuring,hubert1985comparing,morey1984measurement}. Ideally one would like to get  $\theta - \theta_0$ with their true values. 
Under our multinomial model this quantity is 
\[ 
\theta - \theta_0 = \sum_{k\ell}\pi_{k\ell}^2 - \sum_{k\ell}\pi_{k.}^2\pi_{.\ell}^2
\]
which is equal to zero under $\HypNull$ (independence of the classifications), that is, when $\pik \pil = \pikl$ for all $k,\ell$. In practice, one can only estimate the quantities $\theta - \theta_0$ from observed classifications. Our goal is therefore to get an unbiased estimator of $\theta - \theta_0$. 

The $MRI$ being by definition an unbiased estimator of $\theta$, we only need an unbiased estimator of $\theta_0$, that is $\sum_{k\ell} \pi_{k.}^2 \pi_{.\ell}^2$. However, under the alternative $\HypAlt$ (i.e. when the compared classifications are not independent, the most natural case), deriving an unbiased estimator of $\theta_0$ is trickier and depends on the model assumption. \cite{morey1984measurement} proposed a plug-in estimator for the multinomial model, but as pointed out by \cite{hubert1985comparing, steinley2004properties,steinley2018note}, they made errors in their calculations. In the next section we continue their work by proposing an unbiased estimator for $\theta_0$. We also show that the hypergeometric estimator of \cite{hubert1985comparing} for $\theta_0$, used as correction in the "traditional" $ARI$,  is biased under our multinomial $\HypAlt$.

\paragraph*{A new Adjusted  Rand Index.} We now define our own adjusted version of the $MRI$ that we denote $MARI$:
\begin{equation}
    \label{eq:MARI}    
    MARI = \widehat{\theta} - \widehat{\theta}_0. 
\end{equation}
with 
$$\widehat{\theta} = \sum_{\P} C_{ij}^1 C_{ij}^2 \Big/ {n \choose 2} \qquad \widehat{\theta}_0 = \sum_{\Q} C_{ij}^1 C_{i'j'}^2 \Big/ 6 {n \choose 4}. $$

and its observed value, 

  $$MARI^{obs} = \sum_{\P} c_{ij}^1 c_{ij}^2 \Big/ {n \choose 2} - \sum_{\Q} c_{ij}^1 c_{i'j'}^2 \Big/ 6 {n \choose 4}.$$

where we recall that $c_{ij}^1$ and $c_{ij}^2$ are the observed counterparts of $C^1_{ij}, C^2_{ij}$ and $\P, \Q$ are defined in  Section \ref{sec:quadruplets}.

\begin{lemma}
Under the multinomial model, the $MARI$ is unbiased, that is,
$$\mathbb{E}(MARI) = \theta - \theta_0.$$
\end{lemma}

\begin{proof}
    The proof is straightforward using Lemma \ref{lemma:quad}.
\end{proof}

\paragraph*{Computing the $MARI$ from a contingency table.} In practice, the comparison of two classifications is given as a contingency table as Table \ref{tab:n_kl Contigency}, and we thus need a formulation of the $MARI$ defined in \eqref{eq:MARI} as a function of $n_{k\ell}$.

We already gave in \eqref{eq:formRIM} an expression of $\widehat{\theta}$  
as a function of $n_{k\ell}$. As we will see, $\widehat{\theta}_0$ can as well 
be computed from the $\nkl$ contingency Table \ref{tab:n_kl Contigency} even if summing over all elements of $\Q$ rather than $\P$ is a bit less straightforward. To get $\sum_{ \Q } c^1_{ij} c^2_{i'j'}$, we will use the term $\sum_{k\ell}\nk^2\nl^2$ from which we will, as a direct result of Definition \eqref{eq:cVarCounts}, derive the $(\sum_{\P}\cvarREALONE)(\sum_{\P} \cVarREALTWOPrime)$ terms. These latter can be decomposed as follows:

\begin{equation}
\label{eqn:allNtupltsSimilar}
(\sum_{\P} \cvarREALONE)(\sum_{\P} \cVarREALTWOPrime) = \sum_{ \P}\VarREALPairs + \sum_{\T} c^1_{ij} c^2_{ij'} + \sum_{ \Q }\VarREALQuadruplets.    
\end{equation}
It is then sufficient to subtract the terms of $\P$ and $\T$ from the left side of Equation \eqref{eqn:allNtupltsSimilar} to get $\sum_{ \Q } c^1_{ij} c^2_{i'j'}$.
 All terms summing over $\P$ are easy to recover (see Definition~\ref{eq:cVarCounts}).
However, the terms involving elements of $\T$ are more tedious to obtain and are derived in Lemma~\ref{lemma:counttriplets}. The terms of $\Q$ derived in Lemma~\ref{lemma:countquadruplets}.

\begin{lemma}\label{lemma:counttriplets} We have the following expression of $\sum_{\T} c^1_{ij} c^2_{ij'}$ in terms of $n_{k\ell}$:
 \begin{eqnarray*}
\sum_{\T} c^1_{ij} c^2_{i'j}   =  
 2n + \sum_{k,\ell} n_{k.}n_{k\ell}n_{.\ell} - \sum_{k, \ell} n_{k\ell}^2 - \sum_k n_{k.}^2 - \sum_{\ell} n_{.\ell}^2
 \end{eqnarray*}
\end{lemma}{}

\begin{proof}

We need to consider all $i$ in $\{1, ..., n\}$. Assuming for now that $i$ is in classes $(k,\ell)$, that is  $c^1_i=k$ and $c^2_i=\ell$, let us consider all $j, j'$ such that $c^1_{ij}c^2_{ij'}=1$. 
The term $c^1_{ij}c^2_{ij'}$ is equal to one if $c^1_j=k$ and $c^2_{j'} = \ell$. We then get different scenarios according to whether $c^1_{j'} = k$ or not and whether $c^2_{j} = \ell$.  Those scenarios are enumerated in Table \ref{tab:scenario}.

\begin{table}[ht!]
\centering
\caption{Four scenarios to be considered for $j$ and $j'$ in the calculation of the terms in$\sum_{\T} c^1_{ij}c^2_{ij'}$ when $i$ is in class $(k, \ell).$}\label{tab:scenario} 
\begin{tabular}{|c|c|c|}
\hline
& $j$ in $\ell$  & $j$  not in $\ell$ \\ \hline
$j'$ in $k$ & $(n_{k\ell}-1)(n_{k\ell}-2)$     &     $(n_{k\ell} - 1)(n_{.\ell} - n_{k\ell}) $ 
\\ \hline
$j'$ not in $k$  & $(n_{k.} - n_{k\ell})(n_{k\ell}-1)$     &     $(n_{k.} - n_{k\ell})(n_{.\ell} - n_{k\ell})$  \\ \hline
\end{tabular}

\end{table}

Summing all terms of Table \ref{tab:scenario} we get $n_{k.}n_{.\ell} + 2 - n_{k\ell} - n_{k.} - n_{.\ell}$. To account for all $i$ belonging to class $(k, \ell)$ we then multiply by $n_{k\ell}$. Eventually we sum over all $k, \ell$ to recover
\begin{eqnarray*}
\sum_{\T} c^1_{ij} c^2_{ij'} & = & \sum_{k, \ell} n_{k\ell}  ( 2 + n_{k.}n_{.\ell} - n_{k\ell} - n_{k.} - n_{.\ell} ) \\
& = & 2n + \sum_{k,\ell} n_{k.}n_{k\ell}n_{.\ell} - \sum_{k, \ell} n_{k\ell}^2 - \sum_k n_{k.}^2 - \sum_{\ell} n_{.\ell}^2
 \end{eqnarray*}
 \end{proof}

\begin{lemma}
\label{lemma:countquadruplets}
We have the following expression of $\sum_{\Q} c^1_{ij} c^2_{i'j}$ in terms of $n_{k\ell}$:

\begin{multline*}
        \sum_{\Q} c^1_{ij} c^2_{i'j'}  
         = \\ \bigg[\sum_{k\ell} \nk^2\nl^2 -  \bigg(4  \sum_{k\ell}{\nkl\choose{2}} +  4(2n + \sum_{k,\ell} n_{k.}n_{k\ell}n_{.\ell} - \sum_{k, \ell} n_{k\ell}^2 - \sum_k n_{k.}^2 - \sum_{\ell} n_{.\ell}^2) \\
        \qquad\qquad\qquad + 2n \big(\sum_{k} {\nk \choose{2}} + \sum_{\ell} {\nk \choose{2}} \big) + n^2\bigg)\bigg]\bigg/4
    \end{multline*}
\end{lemma}

\begin{proof} 
From Equation \eqref{eq:cVarCounts} we can derive $\sum_{k\ell}\nk^2\nl^2$ as a function of $\sum_{\P\times\P} c^1_{ij} c^2_{i'j'}$ and $n$, since, $\sum_k \nk^2 = n + 2 \sum_\P \cvarREALONE$ and $\sum_\ell \nl^2 = n + 2 \sum_{\P}\cVarREALTWO$ with,

\begin{equation}
    \begin{split}
        \sum_{k\ell}\nk^2 \nl^2 & =  (2 \sum_{i < j} \cvarREALONE + n)(2 \sum_{i' < j'} \cVarREALTWOPrime + n)  \\
               & = 4  \sum_{\P\times\P} \VarREALQuadruplets + 2n \bigg(\sum_{\P} \cvarREALONE + \sum_{\P} \cVarREALTWOPrime \bigg) + n ^2 \\
    \end{split}
\end{equation}

Using equation \eqref{eqn:allNtupltsSimilar}, we decompose $\sum_{\P\times\P} \VarREALQuadruplets$ into terms of $\P$, $\T$ and $\Q$ and get, 

\begin{multline*}
    \sum_\Q \VarREALQuadruplets = \bigg[\sum_{k\ell} \nk^2\nl^2 -  \bigg(4  \sum_{\P}\VarREALPairs +  4\sum_{\T}\VarREALTriplets  + 2n \big(\sum_{\P} \cvarREALONE + \sum_{\P} \cVarREALTWO \big) + n ^2\bigg)\bigg]\bigg/4\\
     \quad\quad\quad\quad = \bigg[\sum_{k\ell} \nk^2\nl^2 -  \bigg(4  \sum_{k\ell}{\nkl\choose{2}} +  4(2n + \sum_{k,\ell} n_{k.}n_{k\ell}n_{.\ell} - \sum_{k, \ell} n_{k\ell}^2 - \sum_k n_{k.}^2 - \sum_{\ell} n_{.\ell}^2)\\
     + 2n \big(\sum_{k} {\nk \choose{2}} + \sum_{\ell} {\nk \choose{2}} \big) + n ^2\bigg)\bigg]\bigg/4
\end{multline*}
\end{proof}

\section{Implementation - package \texttt{aricode}}\label{sec:aricode}

We implemented code for fast  computation of the $MRI$ and its adjusted version the $MARI$, as well as a number of other clustering comparison measures in the \texttt{R/C++} package \texttt{aricode}, which is available on CRAN. 

Computing these measures is  straightforward by means of the whole $K\times L$ contingency table. However, the time and space complexity  is in $\mathcal{O}(n + KL)$, which is somewhat inefficient when $K$ and $L$ are large. Our implementation in \texttt{aricode} is in $\mathcal{O}(n)$: the key idea  is that, given $n$ observations, at most $n$ elements  of the $n_{k\ell}$ contingency matrix can be non zero. To recover these non zero elements one can proceed in two simple steps: first, all observations are sorted in lexicographical order in terms of their first and second cluster index. This can be done in $\mathcal{O}(n)$ using  \texttt{bucket sort} \citep{Cormen2001introduction} or \texttt{radix sort} (as implemented in \texttt{R} \citep{Rproject}). Note that once the observations are sorted, all $i$ that are in clusters $k$ and $\ell$ are one after the other in the data table. Thus, in a second step  \texttt{aricode} counts all non zero $n_{kl}$ in a single path over the data table. Internally this is done using \texttt{Rcpp} \citep{eddelbuettel2011rcpp}.

In Figure \ref{fig:TimingAricodeVsMclust} we compare our implementation of the standard ARI with the implementation of \texttt{mclust} \citep{mclust} (that uses the whole contingency table). As can be noted, the cost of the latter can be prohibitive for large vectors.

\begin{figure}[htbp!]
    \centering
    \includegraphics[scale=0.5]{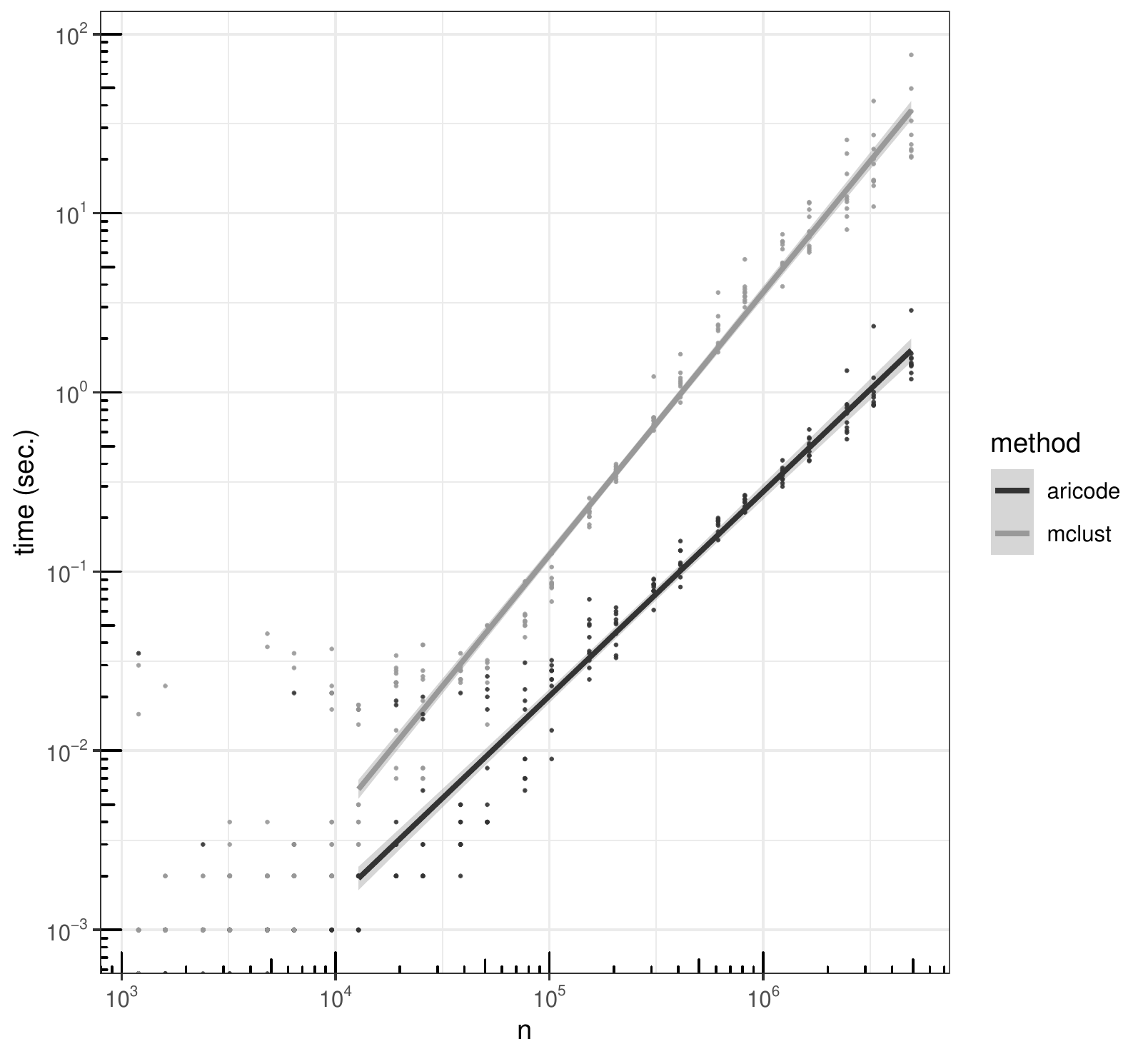}
    \caption{Timings comparing the cost of computing the ARI with \texttt{aricode} or with the commonly used function \texttt{adjustedRandIndex} of the \texttt{mclust} package.}
  \label{fig:TimingAricodeVsMclust}
\end{figure}

\section{\citeauthor{hubert1985comparing}'s ARI}\label{sec:HubertARI}

In this section we study the expectation of the 'standard' $RI$ of \cite{brennan1974measuring} (by contrast with our $MRI$); the expression of which results from the hypergeometric model. This expression was used by \cite{hubert1985comparing} for adjusting the $RI$ and producing the usual $ARI$.  We study this expected value 
when the expectation corresponds to the multinomial distribution. We show that this estimator is biased in general under 
the alternative hypothesis, that is, when the two compared clusterings are not independent.

\subsection{Expectation of \citeauthor{hubert1985comparing}'s $ARI$}\label{sec:HubertARIExp}

Consider the observed value of the $ARI$ proposed by \cite{brennan1974measuring, hubert1985comparing}: in order to analyse this quantity in our multinomial setup, we first give its definition in terms of $c^1_{ij}$ and $c^2_{ij}$, that is
\begin{eqnarray*}
\ARIOld^{\text{obs}}  & = 
  & \frac{2}{{n\choose{2}}} \sum^{KL}_{kl} {n_{kl}\choose{2}} -\frac{2}{{n\choose{2}}^2} \sum^{KL}_{kl}{n_{k.}\choose{2}}{n_{.l}\choose{2}} \\
& = & \frac{2}{{n\choose{2}}} \sum_{\P} c^1_{ij}c^2_{ij} -\frac{2}{{n\choose{2}}^2} \sum_{\P} c^1_{ij} \sum_{\P} c^2_{ij},
\end{eqnarray*}
where we recall that $c^1_{ij}$,$c^2_{ij}$ are realisations of the Bernoulli variables $C^1_{ij}, C^2_{ij}$. In a probabilistic perspective, we consider the $ARI$ as a random variable:
\begin{equation}
ARI = \underbrace{\frac{2}{{n\choose{2}}} \sum_{\P} C^1_{ij}C^2_{ij}}_{\widehat{\theta}^{RI}} -\underbrace{\frac{2}{{n\choose{2}}^2} \sum_{\P} C^1_{ij} \sum_{\P} C^2_{ij}}_{\widehat{\theta}^{RI}_0},
\end{equation}
where, as for the $MRI$, we ignored the marginal terms in our definitions of  $\widehat{\theta}^{RI}$ and  $\widehat{\theta}^{RI}_0$ that cancel in the $ARI$. We now claim the following proposition.

\begin{prop}
Under the multinomial model we have

$$
\mathbb{E}(ARI) = \mathbb{E}(\hat{\theta}^{RI}) - \mathbb{E}(\hat{\theta}^{RI}_{0}), $$
with 
\begin{gather*}
\mathbb{E}(\hat{\theta}^{RI}) = 2\sum_{k\ell}^{KL}\pikl^2 \text{ and }\\
\mathbb{E}(\hat{\theta}^{RI}_{0}) = \frac{2}{{n\choose{2}}^2}\bigg[{n\choose 2}\sum_{k\ell}^{KL}\pikl^2 + n(n-1)(n-2)\sum_{k\ell}^{KL}\pikl\pik\pil + 6{n\choose 4} \sum_{k\ell}^{KL}\pik^2\pil^2 \bigg]
\end{gather*}

Assuming we are under the null this simplifies so that $\mathbb{E}_{\HypNull}(\ARIOld) = 0$.

\end{prop}

\begin{proof} Using Lemma~\ref{lemma:pairs}, we have

$$\mathbb{E}(\sum_{\P} C^1_{ij} C^2_{ij}) = {n\choose{2}}\sum_{k\ell}^{KL}\pikl^2.$$

Using Definition~\eqref{eqn:allNtupltsSimilar} and Lemmas \ref{lemma:pairs}, \ref{lemma:triplet}, \ref{lemma:quad} we obtain

$$\mathbb{E}(\sum_{\P} C^1_{ij} \sum_{\P} C^2_{ij}) = {n\choose 2}\sum_{k\ell}^{KL}\pikl^2 + n(n-1)(n-2)\sum_{k\ell}^{KL}\pikl\pik\pil + 6{n\choose 4} \sum_{k\ell}^{KL}\pik^2\pil^2. $$

Under the null we have $\pikl^2=\pik^2\pil^2$ and we get

$$\mathbb{E}_\HypNull(\sum_{\P} C^1_{ij} C^2_{ij}) = {n\choose 2}\sum_{k\ell}^{KL}\pik^2\pil^2$$

\begin{eqnarray*}
\mathbb{E}_\HypNull(\sum_{\P} C^1_{ij} \sum_{\P} C^2_{ij}) & = & \sum_{k\ell}^{KL}\pik^2\pil^2\bigg[{n\choose 2} + n(n-1)(n-2) + 6{n\choose 4} \bigg] \\
&=& {n\choose 2}^2\sum_{k\ell}^{KL}\pik^2\pil^2.
\end{eqnarray*}

The expectations $\mathbb{E}(\hat{\theta}^{RI})$ and $\mathbb{E}(\hat{\theta}^{RI}_0)$ are obtained by scaling respectively with $2/{n\choose{2}}$ and $2/{n\choose{2}}^2$;  $\mathbb{E}(ARI)$ is their difference.

\end{proof}

From these results we conclude that \citeauthor{hubert1985comparing}'s $ARI$ is biased  under the multinomial model in general, since the term used for the adjustment is biased as $\mathbb{E}(\widehat{\theta}^{RI}_0) \neq \theta^{RI}_0$. Note, however, that this estimator is not biased under the null $\HypNull$.

\subsection{Study of the bias \citeauthor{hubert1985comparing} 's $ARI$}\label{section:bias}

The quantity that we study in this section is
\begin{equation*}
  \label{eqn:HubertBias}
  \begin{split}
\text{bias}_n(\theta^{RI}_0) & = \theta^{RI}_0 - \mathbb{E}(\widehat{\theta}^{RI}_0)\\
& = \sum_{k,\ell}^{K,L}\pi_{k.}^2\pi_{.\ell}^2 - \Big[{n\choose{2}} \sum_{kl}^{K,L}\pi_{k\ell}^2 +  6{n\choose{3}}\sum_{k\ell}^{K,L}\pi_{k\ell}\pi_{.\ell}\pi_{k.} + 6{n\choose{4}}\sum_{k\ell}^{K,L}\pi_{k.}^2\pi_{.\ell}^2 \Big] \bigg/{n\choose{2}}^2
    \end{split}
\end{equation*}

\paragraph{Bias disappear when $n$ goes to infinity.}

The bias can be rewritten as

\begin{equation*}
\label{eq:biasRewritten}
\begin{split}
\text{bias}_n(\theta^{RI}_0) = \frac{4n - 6}{n(n-1)}\sum_{k,\ell}^{K,L}\pi_{k.}^2\pi_{.\ell}^2 -  \frac{2}{n(n-1)}\sum_{kl}^{K,L}\pi_{k\ell}^2 -  \frac{4(n-2)}{n(n-1)}\sum_{k\ell}^{K,L}\pi_{k\ell}\pi_{.\ell}\pi_{k.} 
    \end{split}
\end{equation*}

From this expression we get 

\begin{lemma}\label{lemma:boundbias}
\begin{equation*}
    |\text{bias}_n(\theta^{RI}_0)| \leq \frac{8}{n}
\end{equation*}
\begin{equation*}
 |\text{bias}_n(\theta^{RI}_0)| = O(1/n), \quad \text{and} \quad
 \lim_{n \to  +\infty}\text{bias}_n(\theta^{RI}_0) = 0.  
\end{equation*}
\end{lemma}
\begin{proof}
As seen in Equation~\eqref{eq:biasRewritten}, the bias consist of three terms. The absolute value of the sum of these three terms is bounded by the sum of their absolute values. Then, using that $\sum_{k, \ell} \pi_{k\ell} = 1$ and all $\pi_{k\ell} \geq 0$, we bound $\sum_{k,\ell}\pi_{k.}^2\pi_{.\ell}^2$, $\sum_{kl}\pi_{k\ell}^2$ and $\sum_{k\ell}\pi_{k\ell}\pi_{.\ell}\pi_{k.}$ by $1$ and we get
$|\text{bias}_n(\theta^{RI}_0)| \leq \frac{4(2n-3)}{n(n-1)}$. We have, $(2n-3) < 2(n-1)$ and the result follows.

\end{proof}



\paragraph{Empirical bias.} In the case of independence the bias is zero. In the case of dependence, using Lemma \ref{lemma:boundbias} we get that the bias is smaller than $0.04$ for $n$ larger than $200$.
Following the work of \cite{steinley2018note}, we study the importance of the difference empirically for small value of $n$ in the next paragraph. In summary for $n$ larger than $64$ we observe a small bias, typically smaller than $10^{-2}$.
For smaller values of $n$ the bias can be larger.

\paragraph*{Simulation setting.} We study the evolution of the bias by comparing two classifications with equal number of groups ($K = L$), with values varying in $K \in \{2, 4, 8, 16, 32, 64, 128 \}$ and a growing number of individuals. For drawing the two compared classifications under the multinomial model, see Table~\ref{tab:Pi Contigency}. We consider three scenarios described below where  we tune the level of difficulty by controlling the balance between group sizes with the parameters $\epsilon$.

\begin{description}
    \item{Scenario 1.} In the first scenario we investigate a $\pi_{kl}$ distribution with a disproportionate diagonal. All other entries being null.

\begin{equation*}
\pi_{k\ell} = 
\begin{pmatrix}
1-\epsilon & 0 & \cdots & 0 \\
0 & \frac{\epsilon}{K-1} & \cdots & 0 \\
\vdots  & \vdots  & \ddots & \vdots  \\
0 & 0 & \cdots & \frac{\epsilon}{K-1} 
\end{pmatrix}
\end{equation*}

\item{Scenario 2.} In the second scenario we investigate a $\pi_{kl}$ distribution with a proportional diagonal and extra diagonal dependency. All other entries being null.

\begin{equation*}
\pi_{k\ell} = 
\begin{pmatrix}
(1-\epsilon)/K & \epsilon/K & \cdots & 0 \\
0 & (1-\epsilon)/K & \cdots & 0 \\
\vdots  & \vdots  & \ddots & \vdots  \\
\epsilon/K & 0 & \cdots & (1-\epsilon)/K 
\end{pmatrix}
\end{equation*}

\item{Scenario 3.} In the third scenario we investigate a $\pi_{kl}$ distribution with one line and one column being disproportional and all other entries being null.

\begin{equation*}
\pi_{k\ell} = 
\begin{pmatrix}
1-\epsilon & \frac{\epsilon}{K+L-2} & \cdots & \frac{\epsilon}{K+L-2} \\
\frac{\epsilon}{K+L-2} & 0 & \cdots & 0 \\
\vdots  & \vdots  & \ddots & \vdots  \\
\frac{\epsilon}{K+L-2} & 0 & \cdots & 0 
\end{pmatrix}
\end{equation*}

\end{description} 

\paragraph{Results.} The results are shown in Figure~\ref{fig:BiasHubert} where the bias is shown in its absolute value with $\log_2/\log_{10}$ scales. For the different scenarios, the parameter of imbalanceness $\epsilon$,is fixed to $0.3$ and $0.8$. 

In the different scenarios, the bias remains moderates for most values of $K$ and $n$. When the number of individuals is small however, the difference turns to be more important and using the $(A)RI$ lead to misguiding conclusions. 

\begin{figure}[H]
\hspace*{-0.5cm}
  \includegraphics[scale=0.45]{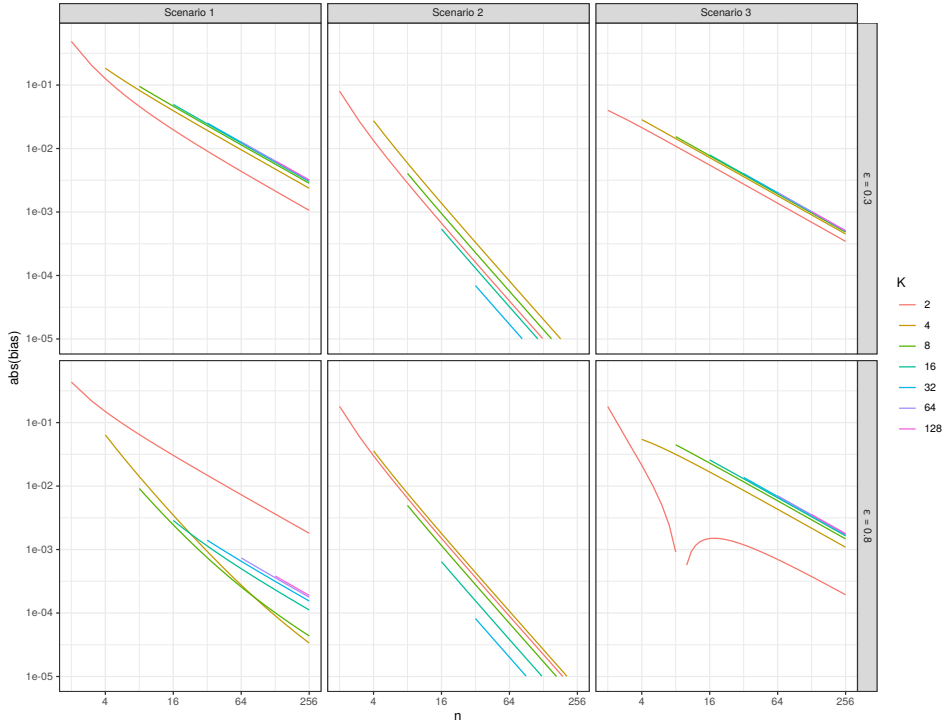}
  \caption{\citeauthor{hubert1985comparing}'s ARI bias for different scenarios of $\pikl$-distribution}
  \label{fig:BiasHubert}
\end{figure}

\section{Conclusion}

As a conclusion, we argue that one should always prefer our $M(A)RI$ to the $(A)RI.$ There are four main reasons for this.

\begin{itemize}
    \item The adjustment of the $RI$ is based on a hypergeometric distribution which is unsatisfying from a modeling perspective. In particular, it forces the size of the clusters to be the same and it ignores randomness of the sampling (see the introduction). A multinomial model of the $MARI$ does not force the size of the clusters and properly model randomness. Furthermore, the model easily extends to the dependant case.
    \item The difference between the $ARI$ and $MARI$ can be large for small $n$ but essentially vanish for large $n$  (see Section \ref{section:bias}).
    \item The $M(A)RI$ can be computed just as fast as the $(A)RI$ in only $O(n)$ rather than $O(n + KL)$ using our \texttt{aricode} package.
    \item The $M(A)RI$ does not take into account pairs coherent by difference which -- as argued in Section \ref{sec:newdefRI} -- unnecessarily complexify the analysis and interpretation of the $(A)RI$.
\end{itemize}

\begin{acknowledgements}
This work is supported by (1) allocations doctorales sur domaines ciblés  (ARDoc) de la Région île de France and (2) an ATIGE grant from Genopole. The IPS2 benefits from the support of the LabExSaclay Plant Sciences-SPS.
\end{acknowledgements}

%
\section*{Conflict of interest}
We declare that we have no conflict of interest.



\bibliographystyle{plainnat}
\bibliography{refs}

\eject

%
%

\end{document}